\documentclass[11pt,a4paper,twoside]{article}

\usepackage[utf8]{inputenc}
\usepackage[margin=2cm]{geometry}
\usepackage[shortlabels]{enumitem}
\usepackage[page,toc]{appendix}
\usepackage{amsmath,amssymb,graphicx,xcolor,subcaption,float,amsthm,mathrsfs,
mathtools,bbold,framed,url,amscd,soul,mathabx,tikz,tikz-cd}
\usepackage[color=yellow]{todonotes}
\usepackage{comment}
\setstcolor{red}

\usepackage{pgfplots}

\usepackage{hyperref}
\hypersetup{
colorlinks=true,
}
\numberwithin{equation}{section}

\newcommand{\nocontentsline}[3]{}
\newcommand{\tocless}[2]{\bgroup\let\addcontentsline=\nocontentsline#1{#2}\egroup}

\newtheorem{theorem}{Theorem}
\newtheorem{corollary}[theorem]{Corollary}

\theoremstyle{definition}

\theoremstyle{remark}
\newtheorem{remark}{Remark}[section]

\usepackage{fancyhdr}
\emergencystretch=\maxdimen
\newcommand{\bs}[1]{\boldsymbol{#1}}
\newcommand{\wh}[1]{\widehat{#1}}
\newcommand{\wt}[1]{\widetilde{#1}}
\newcommand{\bb}[1]{\mathbb{#1}}

\def\p{{\partial}}

\begin{document}
    \title{Coupling of waves to sea surface currents via horizontal density gradients }
	\author{Darryl D. Holm, Ruiao Hu and Oliver D. Street\footnote{Corresponding author, email: o.street18@imperial.ac.uk}\\
	d.holm@imperial.ac.uk, ruiao.hu15@imperial.ac.uk, o.street18@imperial.ac.uk\\
	Department of Mathematics, Imperial College London \\ SW7 2AZ, London, UK}
	\date{Key words: nonlinear water waves, free surface fluid dynamics, geometric mechanics}
	
	\maketitle

	\begin{abstract}
	    The mathematical models and numerical simulations reported here are motivated by satellite observations of horizontal gradients of sea surface temperature and salinity that are closely coordinated with the slowly varying envelope of the rapidly oscillating waves. This coordination of gradients of fluid material properties with wave envelopes tends to occur when strong horizontal buoyancy gradients are present. The nonlinear models of this coordinated movement presented here may provide future opportunities for the optimal design of satellite imagery that could simultaneously capture the dynamics of both waves and currents directly.
	    
	    The model derived here appears in two levels of approximation: first for rapidly oscillating waves, and then for their slowly varying envelope (SVE) approximation obtained by using the WKB approach. The WKB wave-current-buoyancy interaction model derived here for a free surface with significant horizontal buoyancy gradients indicates that the mechanism for the emergence of these correlations is the ponderomotive force of the slowly varying envelope of rapidly oscillating waves acting on the surface currents via the horizontal buoyancy gradient. In this model, the buoyancy gradient appears explicitly in the WKB wave momentum, which in turn generates density-weighted potential vorticity whenever the buoyancy gradient is not aligned with the wave-envelope gradient.  
	\end{abstract}
	


\section{Introduction}

\subsection{Submesoscale sea surface dynamics}

Capabilities in sea surface observation {have been improving rapidly during the past two decades \cite{SatOprog-A2022}. 
In particular,} new high-resolution satellite observation capabilities are revealing sea surface features {seen for the first time} at \emph{submesoscale} spatial scales of 100 m – 10 km and time scales of hours to weeks. Invariably, the new satellite imagery reveals a plethora of coupled dynamical {surface} phenomena, including currents, spiral filaments, flotsam patterns, jets and fronts, some of which are detected indirectly through gradients of sea surface temperature, salinity or colour, in addition to the imagery \cite{Chapron2020,F-K_etal2022,Gula_etal2022,Rascle2017,Yurovskaya2018}.    

{The new capabilities in sea surface observation are still developing.} For example, the impending Surface Water Ocean Topography (SWOT) mission will map the ocean surface mesoscale sea surface height field, as well as a large fraction of the associated submesoscale field, including buoyancy fronts \cite{Morrow_etal2019}. A sample of this type of submesoscale data taken from \cite{Chapron2020} is shown in figures \ref{fig:snapshot Chapron} and \ref{fig:coherence}.

The coming new age of higher-resolution upper ocean observations will present a formidable array of challenges for the next generation in data management, computational simulation and mathematical modelling.
{This paper will offer a mathematical modelling framework that is flexible enough to admit} uncertainty quantification through stochastic modelling and analysis, applied in concert with high-resolution observations, computational simulations, and stochastic data assimilation for large data sets. { This framework involves decomposing the surface motion into a two-dimensional horizontal flow map representing transport by the current acting on a one-dimensional vertical flow map representing wave-like motion of the elevation. This composition-of-maps modelling framework is described and applied to model sea-surface dynamics in two deterministic examples in section \ref{CoM-approach} of the present paper.}



\begin{figure}[h!]
	\centering
	\includegraphics[width=0.45\textwidth, height=0.45\textwidth]{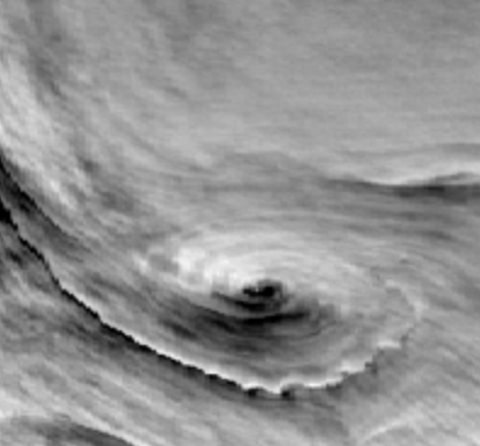}
	\caption{Wave activity in the submesoscale ocean is dynamically complex, as illustrated in this figure showing the zoomed image of a submesoscale sea surface elevation, seen in Envisar MERIS glitter observations. This image shows the wave elevation tracking a cyclonic eddy visible in the sea surface glitter observations. The pixel resolution is 250m. This glitter image demonstrates the complex, highly-coordinated dynamical forms taken in wave-current interaction on the submesoscale sea surface. In particular, notice the instabilities developing in the eddy's outer boundary. Image courtesy of B. Chapron.}
	\label{fig:snapshot Chapron}
\end{figure}
\begin{figure}[h!]
\centering
    \begin{subfigure}[b]{.7\textwidth}
      \centering
      \includegraphics[width=\textwidth]{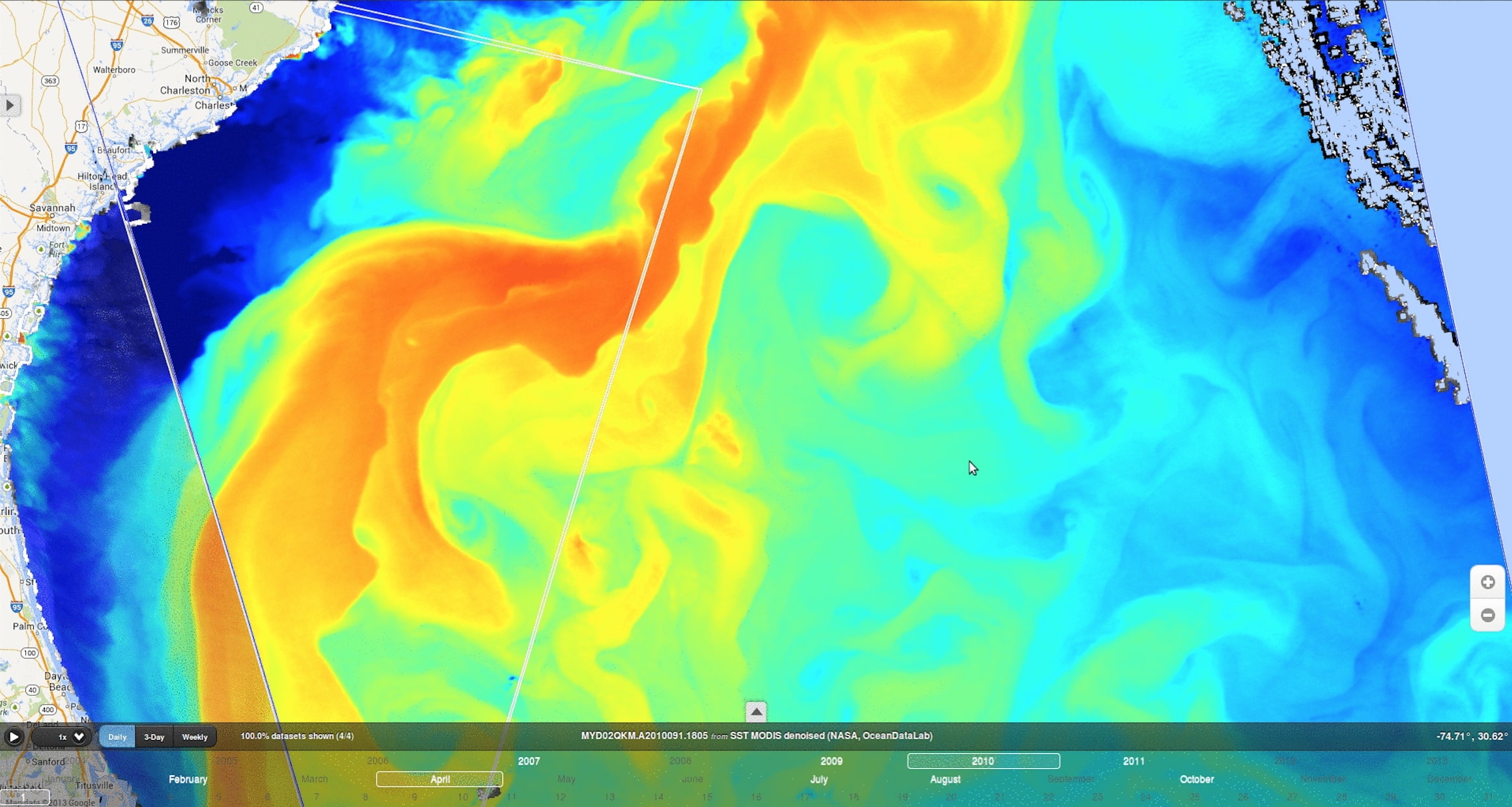}%
      \caption
        {%
          Sea surface temperature near the Gulf Stream, on April 1st 2010, from the Envisat AATSR measurements.%
          \label{fig:SST}%
        }%
    \end{subfigure}
    \begin{subfigure}[b]{.7\textwidth}
      \centering
      \includegraphics[width=\textwidth]{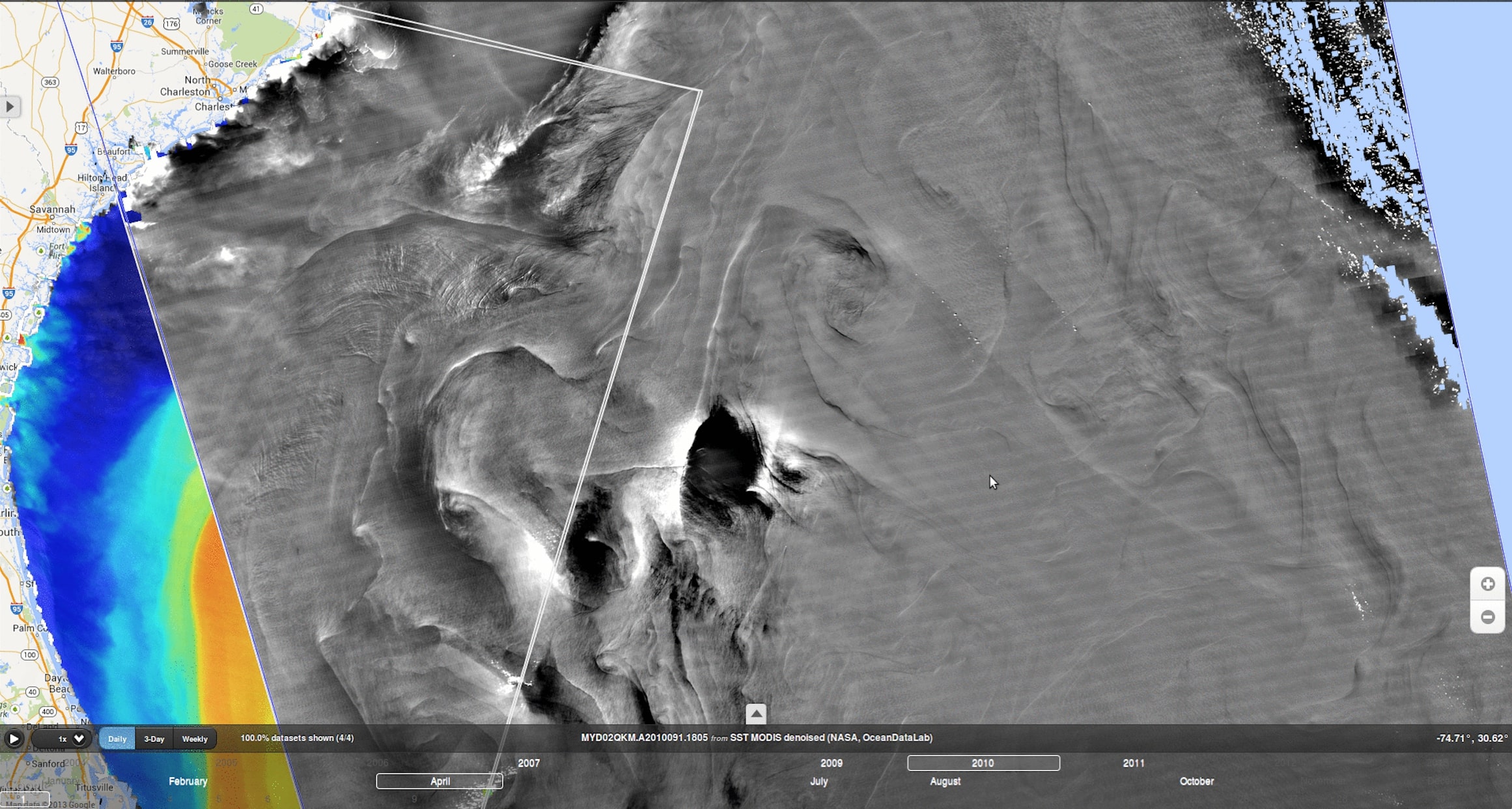}%
      \caption
        {%
          Sea surface glitter contrasts near the Gulf Stream, on April 1st 2010, from the Envisar MERIS observations.%
          \label{fig:glitter}%
        }%
    \end{subfigure}
\caption{Comparison of the two images above demonstrates the emergent coherence between sea surface temperature and the glitter patterns visible from satellite imagery. The thermal fronts visible are dynamic, and sea surface roughness is most obvious along the strongest fronts. Discussions of the interpretation of sun glitter measurements are given in \cite{Chapron2020,Rascle2017,Yurovskaya2018}. Images courtesy of B. Chapron.
}
\label{fig:coherence}
\end{figure}


\textbf{Emergent coherence (EC).} Combining high-resolution thermal data (buoyancy) with glitter data for the wave elevation as in figure \ref{fig:coherence} has recently revealed yet another interesting feature of submesoscale dynamics. Namely, the observed submesoscale data show extremely high correlations of wave, current and thermal properties \cite{Chapron2020}. This emergent spatial-temporal coherence of dynamic and thermal properties presents a significant challenge for dynamical submesoscale modelling. Accepting this challenge, the aim of this paper is to derive a mathematical model of nonlinear sea surface dynamics whose solutions also demonstrate the emergent coherence observed in combining different types of submesoscale data. This paper derives new {\emph{two-dimensional}} equations that show the emergent coherence (EC) seen {in the sea surface features appearing} in figure \ref{fig:coherence}. The EC behaviour produced by the equations {derived here} are demonstrated in figure \ref{fig:snapshot wcifs} which shows a snapshot of the coherence of buoyancy and wave amplitude distributions in the dynamics of divergence-free {two-dimensional} flow acting on free surface {vertical elevation wave} features moving under gravity. In the model equations, the horizontal buoyancy gradients mediate the interactions between the {vertical elevation} waves and {the horizontal} currents. {The equations of motion represent the current as a time-dependent, area-preserving map of the horizontal plane into itself and the waves as the composition of the horizontal flow map with a time-dependent vertical elevation map. Thus, the model involves a dynamical composition of maps (C$\circ$M).
}
\begin{figure}[h!]
	\centering
	    \begin{subfigure}[b]{0.45\textwidth}
			\centering
			\includegraphics[width=\textwidth, height=\textwidth]{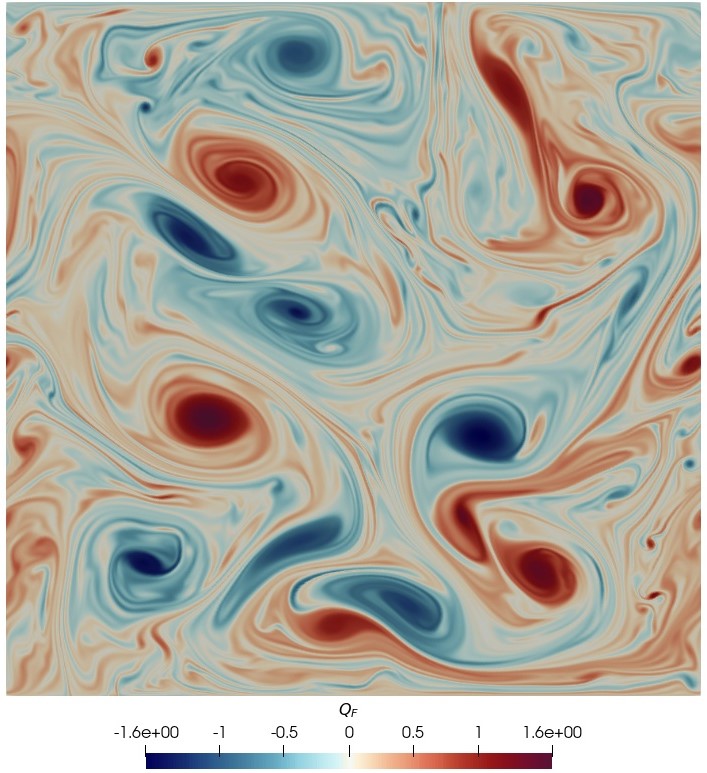}
		\end{subfigure}
		\begin{subfigure}[b]{0.45\textwidth}
			\centering
			\includegraphics[width=\textwidth, height=\textwidth]{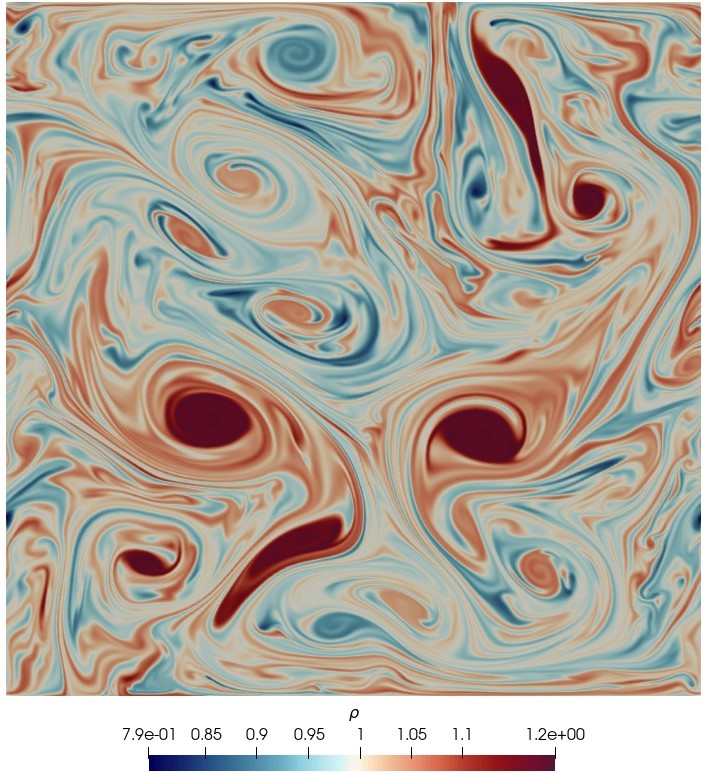}
		\end{subfigure}
		\begin{subfigure}[b]{0.45\textwidth}
			\centering
			\includegraphics[width=\textwidth, height=\textwidth]{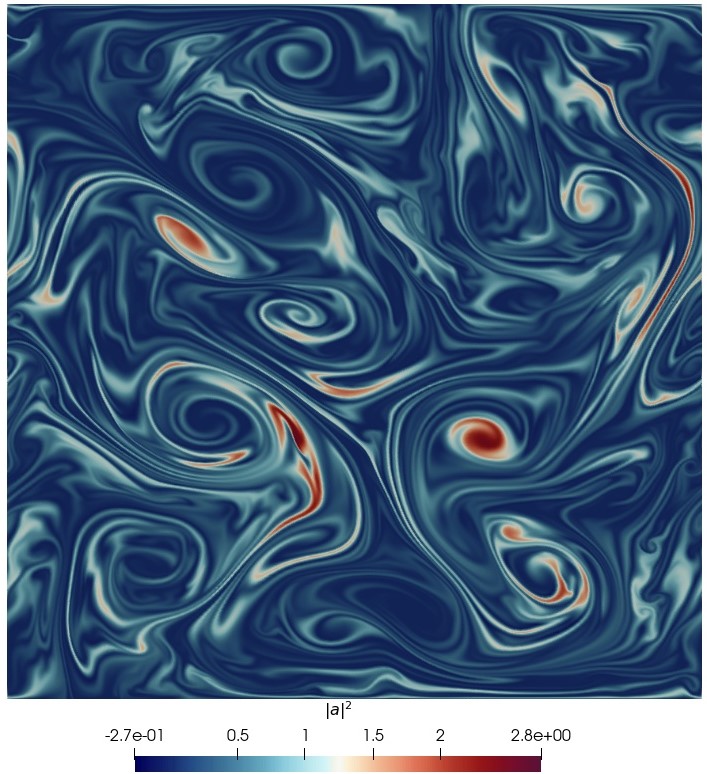}
		\end{subfigure}
		\begin{subfigure}[b]{0.45\textwidth}
			\centering
			\includegraphics[width=\textwidth, height=\textwidth]{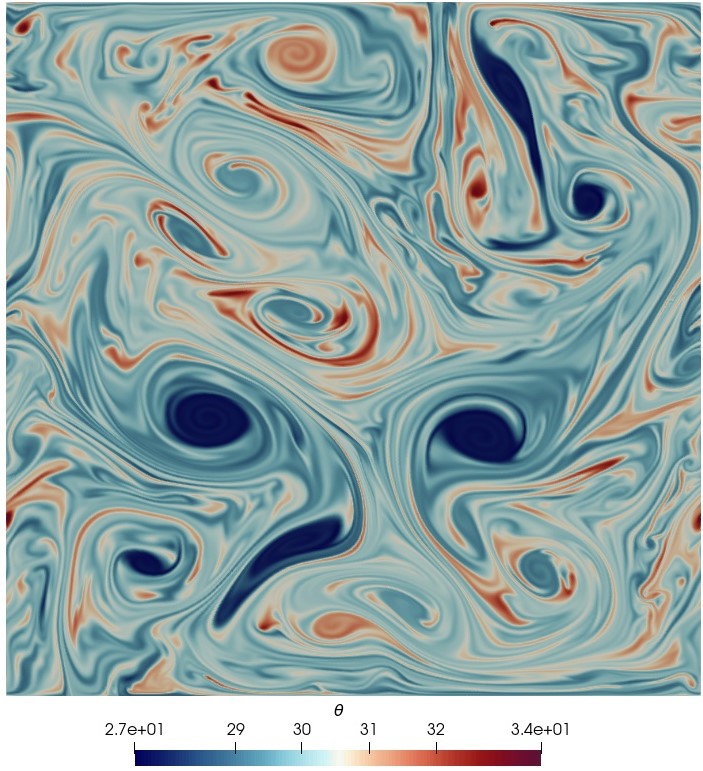}
		\end{subfigure}
		\caption{This is a $512^2$ snapshot of the C$\circ$M equations in the SVE approximation in the potential vorticity form in \eqref{Q-PV-SVE-eqn}. The four panels display the following distributions, modified potential vorticity Q-PV in \eqref{PV-def-SVE} (top left), buoyancy (top right), square of the wave amplitude (bottom left) and wave phase (bottom right) in the numerical simulation of the dynamics of divergence-free flow on a free surface moving under gravity. The simulation began with a spin-up period with zero wave amplitude. After the spin-up period, as explained in section \ref{SimSpecs}, a checker-board pattern of finite wave amplitude with \emph{zero phase} was introduced and the simulation was resumed. The `mixing' of these wave patterns eventually brought them into coherence with the spatial distributions of thermal properties and potential vorticity. These features show an emergent coherence in patterns similar to those seen in the corresponding high-resolution satellite data in figure \ref{fig:coherence}. 
		}
	\label{fig:snapshot wcifs}
\end{figure}

\begin{figure}[h!]
	\centering
	\begin{subfigure}[b]{0.45\textwidth}
		\centering
		\includegraphics[width=\textwidth, height=\textwidth]{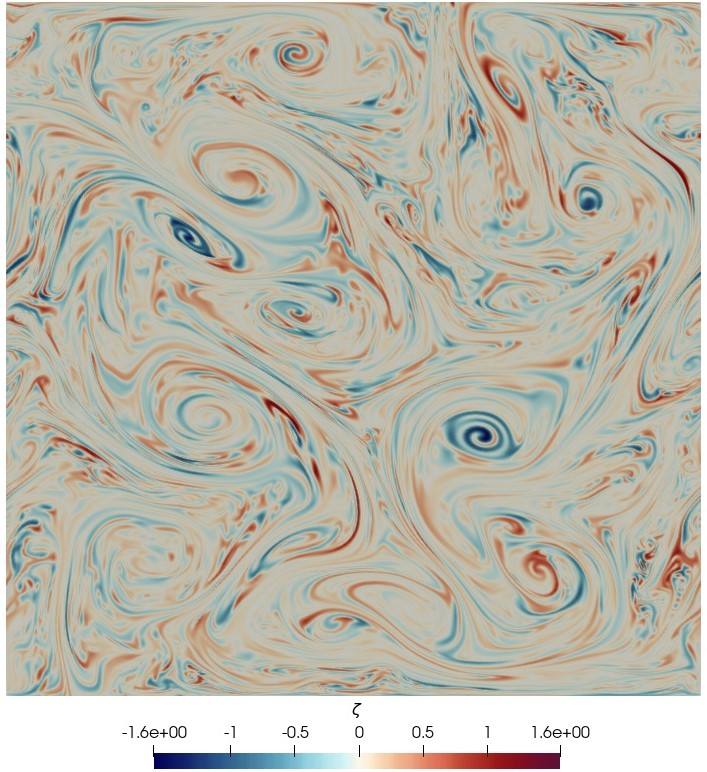}
	\end{subfigure}
	\begin{subfigure}[b]{0.45\textwidth}
		\centering
		\includegraphics[width=\textwidth, height=\textwidth]{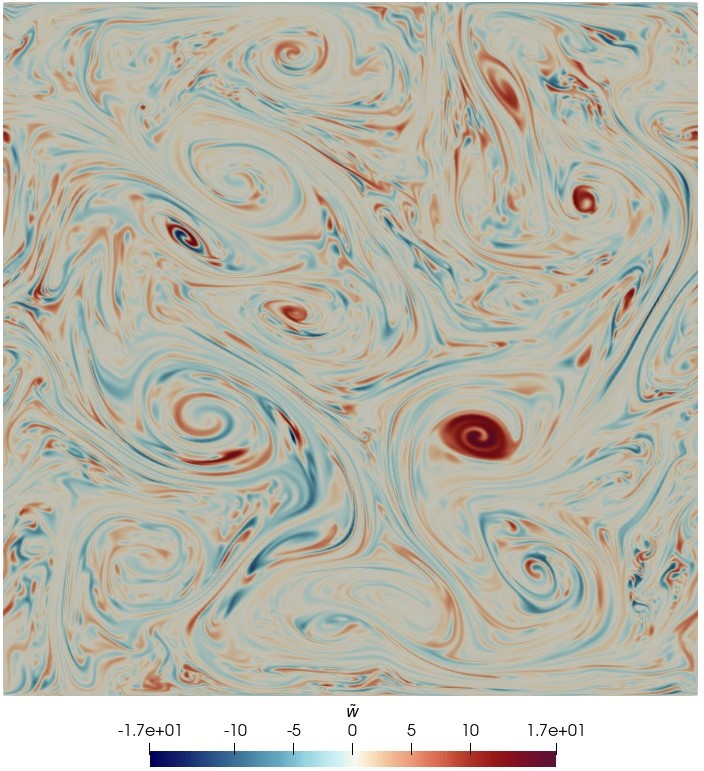}
	\end{subfigure}
	\caption{These $512^2$ snapshots of the C$\circ$M simulation in the vorticity form \eqref{PV-sys-Bdyn} shows the elevation $\zeta$ in the left panel and the density-weighted vertical velocity $\wt{w}$ on the right. The snapshots are taken at the same time and with the same fluid spin-up initial conditions as the snapshots of the simulation of the SVE approximate equations presented in Figure \ref{fig:snapshot wcifs}. Overlaying the two figures demonstrates that the resolved features in the $\zeta$ distribution in this figure of C$\circ$M results are bounded by the SVE wave envelope distribution $|a|^2$ in Figure \ref{fig:snapshot wcifs}.}
	\label{fig:snapshot full wcifs}
\end{figure}


\section[Submesoscale thermal wave-current dynamics]{Submesoscale thermal wave-current dynamics on a free surface}\label{CoM-approach}

{
\subsection{Surface waves as symmetry-breaking features of local force imbalances}
Waves are propagating symmetry-breaking features that signify the response to a local imbalance of forces. 
Thus, from the viewpoint of satellite oceanography, observations of waves -- defined as propagating sea surface elevation features -- signify processes at the surface or below the surface whose presence introduces forces that locally break the symmetry of the surface. The sea surface would otherwise follow the stable global gravitational balance of the geoid, which we regard here as being spherical. Thus, waves arise from a spatially local imbalance of forces in the neighbourhood of a stable equilibrium. The propagating feature of relevance here is the wave elevation, measured as the local departure of the surface level in the direction normal to its equilibrium mean level. The symmetry broken here is the invariance of the sea surface under spatial translations tangent to the equilibrium surface level, also known as the local \emph{horizontal} direction. Hence, from the viewpoint of satellite oceanography, sea surface waves are observed as local vertical displacements of the otherwise horizontal motion of the ocean currents on the sea surface.  From the mathematical modelling viewpoint, sea surface waves are local vertical oscillations of the horizontal surface that are carried along by the horizontal current flow, envisaged as a smooth invertible time-dependent map of the horizontal surface into itself. This is the composition of maps (C$\circ$M) modelling approach for describing the dynamics of horizontal fluid flows (currents) acting on oscillating vertical elevations (waves). Since the surface current velocity, its advected material properties and the wave elevation are all that can be observed in satellite oceanography, the task in three-dimensional ocean modelling for satellite oceanography devolves into determining the dynamical surface features that are produced by the three-dimensional flow processes below the surface arising from e.g., bathymetry, stratification, rotation, Langmuir circulation, and thermal effects such as frontogenesis. The dynamics of the surface signatures of these three-dimensional flow processes, as well as effects of air-sea interactions on the surface need to be interpreted, in order to interpret what satellite oceanography observes. 
}

\subsection{A tale of two maps: currents and waves}\label{sec: CoM intro}

\paragraph{Story line.}
Waves on the surface of the ocean are modelled here as a composition of two smooth invertible maps describing the temporal evolution and advection of two degrees of dynamical freedom interacting at widely separated space-time scales. {In this composition of maps (C$\circ$M) approach,} the waves are regarded as local vertical disturbances that rapidly oscillate as they are swept along by the broad, slowly changing horizontal currents. Thus, the slow current motion is a Lagrangian coordinate for the rapid wave oscillations. This wide separation in space-time scales invokes the classical WKB description. The standard WKB approach seeks a rapidly oscillating wave packet solution whose phase-averaged amplitude possesses a slowly varying envelope (SVE) spatially.  The WKB method is often applied via a variational principle because in a variational setting the phase average naturally leads to an adiabatic invariant known as the wave action density, cf. for example, \cite{Buhler} for a review of the WKB or SVE method in fluid dynamics. Here we will follow the variational approach of \cite{BurbyRuiz2020,GH1996} guided by the classical work of \cite{Voronovich1976,Whitham1967,Whitham2011}.

\paragraph{Submesoscale sea-surface motion: Composition of two time-dependent maps.}
The position and velocity of fluid parcels in motion under gravity on a 2D free surface embedded in $\bb{R}^3$ have both horizontal and vertical components. The corresponding flow maps are denoted as the map $\phi_t :\bb{R}^2\to \bb{R}^2$ for the horizontal current flow,  and  as the composite map $\zeta_t\phi_t $ for the vertical elevation of the waves as a function of time and position in $\bb{R}^2$. The flow lines of these two components of the flow map of a free surface can be written as 
\begin{align*}
\bs{r}_t=\phi_t \bs{r}_0
\quad\hbox{and}\quad
z_t=\zeta_t(\phi_t\bs{r}_0)=:\zeta_t(\bs{r}_t)
\,,\end{align*}
where $\bs{r}_t=(x_t,y_t)\in \bb{R}^2$ is the horizontal position along the flow at time $t$ and $\zeta_t(\bs{r}_t)$ is the vertical elevation at horizontal position $\bs{r}_t$ at time $t$, starting at position $\bs{r}_0$ at time $t=0$. Thus, one may say that the initial position of the flow line, $\bs{r}_0$, is a Lagrangian coordinate for the horizontal motion, and  the horizontal motion is a Lagrangian coordinate for the vertical motion. That is, the `footpoint' at time $t$ of the vertical component of the flow map $\zeta_t$ is located in the horizontal plane along a curve $\phi_t \bs{r}_0$ parameterised by time $t$. Likewise, one can simply say that the wave dynamics is advected, or swept along, by the current dynamics. 

Hence, the corresponding horizontal and vertical components of velocity along a stream line $\bs{r}_t$ in the horizontal plane are defined by, 
\begin{align*}
\frac{d\bs{r}_t}{dt} &= \frac{d}{dt}(\phi_t \bs{r}_0) = \bs{\wh{v}}_t(\phi_t \bs{r}_0) =:  \bs{\wh{v}}_t(\bs{r}_t)
\,,\quad\hbox{so}\quad
\bs{\wh{v}}_t = \frac{d\phi_t}{dt}\phi_t ^{-1}
\quad\hbox{and}\quad
\\
&
\\
\frac{dz_t}{dt}&=:\wh{w}_t(\bs{r}_t)=\frac{d}{dt}\big(\zeta_t(\phi_t\bs{r}_0)\big)
= \partial_t \zeta_t(\bs{r}_t) + \nabla_{\bs{r}} \zeta_t(\bs{r}_t)\cdot  \bs{\wh{v}}_t(\bs{r}_t)
\,.\end{align*}
That is, in the dynamics of free surface flow, the vertical velocity $\wh{w}(\bs{r},t)$ at a given Eulerian point $\bs{r}$ and time $t$ is related to the wave elevation $\zeta(\bs{r},t)$ and horizontal velocity $\bs{\wh{v}}(\bs{r},t)$ at that point by
\begin{align*}
\wh{w}(\bs{r},t) = \partial_t\zeta(\bs{r},t) + \bs{\wh{v}}(\bs{r},t)\cdot\nabla_{\bs{r}}\zeta(\bs{r},t)
\,.\end{align*}
In terms of these fluid variables, {one could propose a} Hamilton's principle for wave-current interaction of a free surface by following \cite{CHS2021} for the variational modelling framework and {applying} \cite{Whitham1967,Craig2016} for the potential energy to find\footnote{In \cite{CHS2021} the potential energy was linear in $\zeta$. This linearity neglected the restoring force due to vertical pressure gradient via Archimedes' principle. Adopting the potential energy quadratic in $\zeta$ regains this restoring force.}
\begin{align}
\begin{split}
0 = \delta S &= \delta\int_a^b \ell(\bs{\wh{v}},\zeta,D,\rho)\,dt
\\&= \delta\int_a^b \int_{\cal D}\bigg( \frac{1}{2}\Big( |\bs{\wh{v}}|^2 
+ \sigma^2\big(\partial_t \zeta + \nabla_{\bs{r}} \zeta \cdot  \bs{\wh{v}}\big)^2 \Big) 
- \frac{\zeta^2}{2Fr^2} 
\bigg) D\rho 
- p(D-1)\,d^2r\,dt
\,.
\end{split}
\label{HP-WCI-A}
\end{align}

To interpret the variational principle proposed in \eqref{HP-WCI-A} we rewrite its Lagrangian as a sum of an Eulerian spatial integral and an integral over material mass elements $d^2r_0=D\rho\, d^2r$ which follow the paths of the horizontal fluid motion, $\bs{r}(\bs{r}_0,t)=\phi_t \bs{r}_0$,
\begin{align}
0 = \delta S &= \delta\int_a^b \int_{\cal D} \frac{D\rho}{2} |\bs{\wh{v}}|^2 
- p(D-1)\,d^2r\,dt
+ 
\delta\int_a^b \int_{{\cal D}_0}
\frac{\sigma^2}{2}\dot{\zeta}^2 - \frac{\zeta^2}{2Fr^2} 
\, d^2r_0\,dt
\,.
\label{HP-WCI-A-Eul-Lag}
\end{align}
Variations of the first summand in \eqref{HP-WCI-A-Eul-Lag} at fixed spatial position $(\bs{r})$ yield the Euler fluid equations for 2D divergence free flow with advected buoyancy, $\rho(\bs{r},t)=\rho(\phi_t\bs{r}_0)=\rho_0(\bs{r}_0)$, 
\begin{align}
\partial_t \bs{\wh{v}} + (\bs{\wh{v}}\cdot\nabla_{\bs{r}})  \bs{\wh{v}} 
= -\,\frac{1}{\rho}\nabla_{\bs{r}} p
\quad\hbox{with}\quad
\nabla_{\bs{r}}\cdot\bs{\wh{v}} =0
\,.\label{Eul-eqn}
\end{align}
Variations of the second summand in \eqref{HP-WCI-A-Eul-Lag} taken at fixed mass element $(\bs{r}_0)$ yield equations for vertical harmonic oscillations of the elevation of each material mass element
\begin{align}
\sigma^2\ddot{\zeta}(\bs{r}_0,t) = \sigma^2\frac{d^2 \zeta}{dt^2}\Big|_{\bs{r}_0} 
= -\, \frac{\zeta(\bs{r}_0,t)}{Fr^2}
\,.\label{Osc-eqn1}
\end{align}
The wave-elevation equation in \eqref{Osc-eqn1} is unrealistic, though, because it implies that fluid mass elements with different labels $(\bs{r}_0)$ would be oscillating in phase and all with the same frequency, as they follow the flow of the Euler fluid equations \eqref{Eul-eqn} for 2D divergence free flow with advected buoyancy. This unrealistic synchronisation and resonance can be removed by including the inertia of each mass element. This can done by including the initial buoyancy of each mass element, as 
\begin{align}
\sigma^2\ddot{\zeta}(\bs{r}_0,t) = \sigma^2\frac{d^2 \zeta}{dt^2}\Big|_{\bs{r}_0} 
= -\,\frac{\rho_{ref}}{\rho_0(\bs{r}_0)} \frac{\zeta(\bs{r}_0,t)}{Fr^2}
\,.\label{Osc-eqn2}
\end{align}
At this point in our reasoning, we have not yet considered the differences in space and time scales between the fluid flow and the wave activity. In what follows, we will use the simple composition-of-maps idea explained here along with estimates of relative space and time scales to investigate the applicability of this class of models. To improve the applicability of the model comprising \eqref{Eul-eqn} and \eqref{Osc-eqn2} for describing the effects of currents on waves, we will derive a related model  in the slowly varying envelope (SVE) approximation. The SVE approximation allows considerations of current and wave dynamics at the same space and time scales. 

The comparisons of the simulated solutions of these C$\circ$M models with the observations in figures \ref{fig:snapshot Chapron}-\ref{fig:snapshot full wcifs} above indicate that these models can indeed produce results that match some aspects of observed features. However, these models are not derived from three dimensional fluid equations. Instead, they are derived from the simple solution ansatz in Hamilton's principle that the vertical elevation of the sea surface wave activity is carried by divergence-free horizontal fluid motion.  The latter assumption is a weakness of the current approach, because it precludes effects of vertical up-welling and down-welling, which are observed to occur along with convergence and divergence of currents \cite{Baylor}. The equations derived here are also not associated with classical surface wave equations such as the nonlinear Schroedinger (NLS) equation, or other celebrated surface wave equations. This departure from the classical water wave literature may be regarded as another weakness of the current approach. 
\color{black}
\paragraph{Estimating parameters $\sigma^2$ and $Fr^2$ for satellite observations.}
The Lagrangian $\ell(\bs{\wh{v}},\zeta,D,\rho)$ in \eqref{HP-WCI-A} represents the dimension-free difference of the kinetic and potential energies, augmented by the incompressibility constraint imposed by the Lagrange multiplier $p$. Two dimension-free parameters ($\sigma^2$ and $Fr^2$) appear in this Hamilton's principle. 
The coefficient $\sigma^2=([H]/[L])^2$ in formula \eqref{HP-WCI-A} is the square of the vertical-to-horizontal aspect ratio. Typically, for satellite observations of submesoscale dynamics one finds 
\[
[H]\approx (3\times 10^{-4} - 3\times 10^{-3})km
\quad\hbox{and}\quad
[L]\approx (10^{-1} - 10)km
\,,\quad\hbox{so}\quad
\sigma^2 \approx 10^{-3} - 10^{-6} \ll1
\]
for the squared aspect ratio $\sigma^2\ll1$ of the height of the waves $[H]$ relative to the breadth $[L]$ of the two-dimensional domain. The squared `Froude number' $Fr^2$ in this regime is estimated by { the square of the ratio of horizontal and vertical frequency scales at the sea surface,}
\begin{align}
Fr^2 := \left(\frac{[V]/[H]}{N }\right)^2 \approx 1 - 10^4 
\,.\label{FrBV-def}
\end{align}
Here, the horizontal velocity on the sea surface is taken as $[V]= (0.1 -1) m/sec$, $[H]=(0.3 - 3)m$. According to \cite{DongF-K_etal2020}, the Brunt-V\"ais\"al\"a buoyancy frequency in the sea surface wave regime is given by $N\approx (10^{-3} - 10^{-4})/sec$.  { The ratio of horizontal and vertical frequency scales at the sea surface in \eqref{FrBV-def} is selected for use later in applying the slowly varying envelope (SVE) wave approximation in section \ref{sec: SVE}. }
Hence, we estimate that the squared product of the `Froude number' and aspect ratio for satellite observations of the sea surface  can reasonably be estimated over the range
\begin{align}
\sigma^2Fr^2 := \left(\frac{[V]}{N [L]}\right)^2 \approx 10^{-3} - 10
\,.\label{FrBV-est}
\end{align}

\paragraph{Modelling the dynamic effects of surface density variations.}

{As mentioned earlier,} the observed oscillations of {sea surface} waves are by no means simultaneous across the whole domain, although {the observations show that} they are indeed coordinated spatially with the buoyancy of the fluid. To correct this solution behaviour, the kinetic energy and potential energy need to be de-synchronised from the buoyancy. 



The dynamic dependence of the wave kinetic energy on the density is physically required. However, to de-synchronise the wave oscillations we can introduce a constant reference density $\rho_{ref}$ into the wave potential energy, by writing
\begin{align}
\frac{\zeta^2}{Fr^2} \to \frac{\rho_{ref} }{\rho} \frac{\zeta^2}{Fr^2}
\quad\hbox{with}\quad
\frac{\rho_{ref} }{\rho} \quad\hbox{of order}\quad O(1)
\,.
 \label{rho-term}
\end{align}
The quantity $\rho_{ref}$ is a constant reference density, and the density ratio $(\rho_{ref}/\rho)=O(1)$ 

The density dependence imposed here is important in the dynamics that follows from Hamilton's principle. 
Substituting the relations in \eqref{rho-term} into Hamilton's principle in equation \eqref{HP-WCI-A} leads to the following dimension-free action integral,
\begin{align}
\begin{split}
0 = \delta S &= \delta\int_a^b \ell(\bs{\wh{v}},\zeta,D,\rho)\,dt
\\&= \delta\int_a^b \int_{\cal D}\bigg( \frac{1}{2}\Big( |\bs{\wh{v}}|^2 
+ \sigma^2\big(\partial_t \zeta + \nabla_{\bs{r}} \zeta \cdot  \bs{\wh{v}}\big)^2 \Big) 
- \frac{\rho_{ref}}{\rho} \frac{\zeta^2}{2Fr^2} 
\bigg) D\rho 
- p(D-1)\,d^2r\,dt
\,.
\end{split}
\label{HP-WCI}
\end{align}

The advected quantities $D(\bs{r},t) d^2r$ and $\rho(\bs{r},t)$ evolve via push-forward by the horizontal flow map, $\phi_t$. For example,  $D_td^2r_t={\phi_t}_*(D_0d^2r_0)$ and $\rho_t={\phi_t}_*\rho_{ref}$ denote, respectively, evolution of the determinant of the Lagrange to Euler map and of the local scalar value of the mass density. Conservation of mass is then expressed as the push-forward relation, $D_t\rho_t d^2r_t={\phi_t}_*(D_0\rho_{ref} d^2r_0)$. The pressure $p$ in \eqref{HP-WCI} acts as a Lagrange multiplier to enforce conservation of area, so that $D_t=1={\phi_t}_*D_0$, and the horizontal flow is incompressible, which implies that the horizontal velocity is divergence-free, i.e., ${\rm div}_{\bs{r}} \bs{\wh{v}}(\bs{r},t)=0$. Taking variations of the action integral \eqref{HP-WCI} yields the following set of equations,
\begin{align}
\begin{split}
	\delta \bs{\wh{v}}:&\quad 
	\frac{\delta \ell}{\delta \bs{\wh{v}}} = D\rho \big( \bs{\wh{v}}\cdot d\bs{r} + \sigma^2\wh{w}\,d\zeta\big) \otimes d^2r\,
	:= D\rho  \bs{V}\cdot d\bs{r}  \otimes d^2r
	\,,\\
	&\hbox{with}\quad 
	\wh{w} = \partial_t\zeta + \bs{\wh{v}}\cdot\nabla_{\bs{r}}\zeta
	\,,\\ 
	\delta \zeta :&\quad 
	\partial_t (\sigma^2D\rho\wh{w}) + \text{div}_{\bs{r}}(\sigma^2D\rho\wh{w}\bs{\wh{v}})
	 - \,D\,{{\frac{\zeta\rho_{ref}}{Fr^2}}} = 0
	\,,\\
	\delta D:&\quad 
	\frac{\delta \ell}{\delta D } 
	= \frac{\rho}{2}\big(|\bs{\wh{v}}|^2+\sigma^2\widehat{w}^2\big)  
	- {{\frac{\rho_{ref}\zeta^2}{2Fr^2} }} - p
	=:\rho{\wt \varpi}  - {\wt p}
	\,,\\
	\delta\rho :&\quad 
	\frac{\delta \ell}{\delta \rho } 
	= \frac{D}{2}\big(|\bs{\wh{v}}|^2+\sigma^2\widehat{w}^2\big) 
	=: D{\wt \varpi} 
	\,,\quad 
	{\wt p} := p + {\frac{ \rho_{ref}\zeta^2 }{2Fr^2}} 
	\,,\\
	\delta p:&\quad 
	D-1 = 0 
	\quad \implies \text{div}_{\bs{r}}\bs{\wh v} =0 
	\,.
\end{split}
\label{var-derivs}
\end{align}
From their definitions as advected quantities, one also knows that $D$ and $\rho$ satisfy
   \begin{align}
   \begin{split}        
        (\p_t+\mathcal{L}_{\bs{\wh v}})(D\,d^2r) = 0 &\Longrightarrow \partial_t D + \text{div}_{\bs{r}}(D\bs{\wh{v}}) 
        =0
        \quad\hbox{with}\quad D=1
        \,,\\ 
        (\p_t+\mathcal{L}_{\bs{\wh v}})\rho = 0 &\Longrightarrow 
        \partial_t\rho  + \bs{\wh{v}}\cdot\nabla_{\bs{r}}\rho = 0       
        \,,
        \end{split}
        \label{AdvecQuants}
    \end{align}
where $\mathcal{L}_{\bs{\wh v}}$ denotes the Lie derivative operation along the horizontal velocity vector field, $\bs{\wh v}$, 
which provides coordinate-free brevity in the notation.
\begin{theorem}[Kelvin-Noether circulation theorem]\label{theoremEP}
Use of the Euler-Poincar\'e (EP) theorem yields the following Kelvin circulation theorem
   \begin{align}
   \frac{d}{dt}\oint_{c(\bs{\wh v})} \big( \bs{\wh{v}}\cdot d\bs{r} + \sigma^2\wh{w}\,d\zeta\big)  
   = - \oint_{c(\bs{\wh v})} \frac{1}{\rho}d{\wt p}\,.
   \label{KNthm}
   \end{align}
\end{theorem}
\begin{proof}
The Euler-Poincar\'e (EP) theorem in this case yields
   \begin{align}
   \begin{split}        
 (\p_t+\mathcal{L}_{\bs{\wh v}})\frac{\delta \ell}{\delta \bs{\wh{v}}} 
 &= 
 \frac{\delta \ell}{\delta D } \diamond D + \frac{\delta \ell}{\delta \rho } \diamond \rho
 :=  D \nabla_{\bs{r}}\frac{\delta \ell}{\delta D } - \frac{\delta \ell}{\delta \rho } \nabla_{\bs{r}} \rho
        \end{split}
        \label{EPeqn}
    \end{align}
Here the diamond $(\,\diamond\,)$ operator is defined by 
   \begin{align}
   \Big\langle  \frac{\delta \ell}{\delta a } \diamond a \,,\, X \Big\rangle_{\mathfrak{X}}
   =:     \Big\langle  \frac{\delta \ell}{\delta a }  \,,\, - \pounds_X a \Big\rangle_V
          \,.\label{EPeqn}
    \end{align}
In addition, $X\in \mathfrak{X}$ is a (smooth) vector field defined on $\bb{R}^2$ and $a\in V$, a vector space of advected quantities, which are here the scalar function, $\rho$, and the areal density $D\,d^2r$. 
Using the advection relations for $D$ and $\rho$ in \eqref{AdvecQuants} and the corresponding variational derivatives in \eqref{var-derivs} simplifies the EP equation in \eqref{EPeqn} to 
\begin{align}
   \begin{split}        
 (\p_t+\mathcal{L}_{\bs{\wh v}})&\Big(\frac{1}{D\rho}\frac{\delta \ell}{\delta \bs{\wh{v}}} \Big)
 =   \frac{1}{\rho} \nabla_{\bs{r}}\frac{\delta \ell}{\delta D } - \frac{1}{D\rho}\frac{\delta \ell}{\delta \rho } \nabla_{\bs{r}} \rho
\,. \\   \hbox{Equation \eqref{var-derivs} then yields} \quad
(\p_t+\mathcal{L}_{\bs{\wh v}})&\big( \bs{\wh{v}}\cdot d\bs{r} + \sigma^2\wh{w}\,d\zeta\big) 
 =
   -\rho^{-1} d{\wt p}  + d{\wt \varpi}
  \,.
    \end{split}\label{EPeqn-simp1}
\end{align}
Inserting the last relation into the following standard relation for the time derivative of a loop integral then completes the proof of equation \eqref{KNthm} appearing in the statement of the theorem,
\begin{align}
    \frac{d}{dt}\oint_{c(\bs{\wh v})}  \big( \bs{\wh{v}}\cdot d\bs{r} + \sigma^2\wh{w}\,d\zeta\big)
    = \oint_{c(\bs{\wh v})} (\p_t+\mathcal{L}_{\bs{\wh v}}) \big( \bs{\wh{v}}\cdot d\bs{r} + \sigma^2\wh{w}\,d\zeta\big)
    = \oint_{c(\bs{\wh v})} -\rho^{-1} d{\wt p}  + d{\wt \varpi}
    \,.
    \label{EPeqn-Kel}
\end{align}
Using the advection relations for $D$ and $\rho$ in \eqref{AdvecQuants} again and combining with the variational relations with respect to $\zeta$ in \eqref{var-derivs} simplifies the $\wh{w}$ and $\zeta$ equations, as follows.
\begin{align}
\begin{split}
	(\p_t+\mathcal{L}_{\bs{\wh v}})\wh{w}
	&= (\p_t+ \bs{\wh v}\cdot \nabla_{\bs{r}})\wh{w}
	 = - \,{\frac{\rho_{ref}}{\sigma^2Fr^2 \rho}\zeta} 
	\,,\\
	(\p_t+\mathcal{L}_{\bs{\wh v}})\zeta &= (\p_t+ \bs{\wh v}\cdot \nabla_{\bs{r}})\zeta = \wh{w} 
	\,.
\end{split}
\label{wave-eqns}
\end{align}
After deriving these equations, one may finally evaluate the constraint $D=1$ imposed by the variation in pressure $p$ to obtain further simplifications.
\end{proof}

\begin{corollary}[Kelvin-Noether circulation theorem for the current]\label{theoremEP-current}
The Kelvin circulation theorem for the current alone is given by,
   \begin{align}
   \frac{d}{dt}\oint_{c(\bs{\wh v})}  \bs{\wh{v}}\cdot d\bs{r}  
   = - \oint_{c(\bs{\wh v})} \frac{1}{\rho}dp - d \frac{|\bs{\wh{v}}|^2}{2} \,.
   \label{KNthm-current}
   \end{align}

\end{corollary}
\begin{proof}
Equation \eqref{KNthm-current} follows by shifting the $\wh{w}d\zeta$ term in equation \eqref{EPeqn-Kel}  to the right-hand side, as
   \begin{align}
   \begin{split}
   \frac{d}{dt}\oint_{c(\bs{\wh v})}  \bs{\wh{v}}\cdot d\bs{r}  
   &= - \oint_{c(\bs{\wh v})} \frac{1}{\rho}d{\wt p} + \sigma^2(\p_t+\mathcal{L}_{\bs{\wh v}})\big(\wh{w}\, d\zeta\big)
   - d{\wt \varpi} 
   \\&= - \oint_{c(\bs{\wh v})} \frac{1}{\rho}d{\wt p} 
   + \sigma^2\big((\p_t+ \bs{\wh v}\cdot \nabla_{\bs{r}})\wh{w} \big)d\zeta + \sigma^2 \wh{w}d\wh{w}
   - d{\wt \varpi} 
   \\&= - \oint_{c(\bs{\wh v})} \frac{1}{\rho}d{\wt p} 
   - \,{\frac{\rho_{ref}}{Fr^2 \rho}\zeta} d\zeta + \sigma^2 \wh{w}d\wh{w}
   - d{\wt \varpi} 
   \\&= - \oint_{c(\bs{\wh v})} \frac{1}{\rho}dp - d \frac{|\bs{\wh{v}}|^2}{2} 
   \\&=: - \oint_{c(\bs{\wh v})} \frac{1}{\rho}dp - d \frac{|\bs{\wh{v}}|^2}{2} \,.
   \end{split}
   \label{KNthm-current1}
   \end{align}

\end{proof}




\begin{remark}[Separation of wave and current circulation.]\label{Remark:nonacceleration}
    The decoupling of the Kelvin-Noether circulation theorem into its wave and current components, leading to the reduction of the current flow to the Euler result in equation \eqref{KNthm-current}, was also observed in \cite{CHS2021}. This behaviour is consistent with the Charney-Drazin `non-acceleration' theorem \cite{CharneyDrazin1961,White1986}. Namely, in certain circumstances, wave activity does not create circulation in the mean current.  
    A modification that allows exchange of circulation between wave (vertical) and current (horizontal) components of the flow was proposed in \cite{CHS2021}. The instabilities observed around the edges of eddies in the satellite imagery shown in figure \ref{fig:snapshot Chapron} suggests that a coupling of this sort may exist at high wave number.
\end{remark}

\begin{remark}
It is clear from equations \eqref{EPeqn-Kel} -- \eqref{KNthm-current} that generation of circulation of the current by the dynamics in equation \eqref{EPeqn-simp1} requires non-zero $\nabla_{\bs{r}}\rho\times\nabla_{\bs{r}}p$.
No current circulation is generated by wave variables in the case of constant buoyancy. 
\end{remark}


\subsection[Thermal potential vorticity (TPV)  dynamics]{Thermal potential vorticity (TPV)  dynamics on a free surface}
The momentum map arising from the variations in \eqref{var-derivs} is given by
\begin{align}
    \frac{1}{D}\frac{\delta \ell}{\delta \bs{\wh v}} = \rho\bs{\wh v}\cdot d\bs{r} + \sigma^2{\rho\wh w}d \zeta\,. \label{momap-wcifs}
\end{align}
As expected from the well-known non-acceleration theorem \cite{{CharneyDrazin1961,White1986}}, the dynamics of the Euler-Poincar\'e equations separate \eqref{EPeqn-simp1} gives the dynamics of the fluid and wave components of the momentum one-form \eqref{momap-wcifs}.
\begin{align}
    \begin{split}
        (\p_t+\mathcal{L}_{\bs{\wh v}}) \Big( \rho \big( \bs{\wh{v}}\cdot d\bs{r} \big)\Big)
 &=
     - dp + \frac{\rho}{2} d\big(|\bs{\wh{v}}|^2\big) 
 \\
 (\p_t+\mathcal{L}_{\bs{\wh v}}) \big( \sigma^2\rho {\wh w}d\zeta \big)
 &= - \,{\frac{ \rho_{ref}}{Fr^2}\zeta} d\zeta + \sigma^2 \rho \wh{w}d\wh{w}
  \,.
    \end{split}\label{eq:momap-dyn wcifs}
\end{align}
The mass-weighted thermal potential vorticity (TPV) also separates into fluid and wave components $Q = Q_F + Q_W$ with following definitions 
\begin{align}
\begin{split}
 Q \,d^2r &= d \Big(\rho\big(\bs{\wh v}\cdot d\bs{r} + \sigma^2{\wh w}d \zeta\big)\Big)
\\&= d\rho\wedge \big(\bs{\wh v}\cdot d\bs{r} + \sigma^2{\wh w}d \zeta\big)
+ \rho\Big(\bs{\wh z}\cdot{\rm curl}\bs{\wh v} + \sigma^2J\big({\wh w}, \zeta\big)\Big) \,d^2r
\\&=
\Big({\rm div}(\rho\nabla\psi) + \sigma^2 J\big(\rho{\wh w},\zeta\big)\Big) \,d^2r
\quad\hbox{when}\quad
\bs{\wh v} =\nabla^\perp\psi \quad \hbox{for} \quad D=1
\,,\\
\hbox{with} \quad Q_F&:={\rm div}(\rho\nabla\psi)\,,\quad Q_W=J\left(\sigma^2{\wt w},\zeta\right)\,.
\end{split}
\label{Q-PV-def}
\end{align}
where buoyancy weighted vertical velocity is defined as $\wt w := \rho \wh w$. The dynamics of of $Q_F\,d^2r$ and $Q_W\,d^2r$ can be computed from \eqref{eq:momap-dyn wcifs} as
\begin{align}
    \begin{split}
        (\p_t+\mathcal{L}_{\bs{\wh v}})  (Q_F \, d^2r) &= \frac{1}{2}d\rho\wedge d(|\bs{\wh{v}}|^2) = \frac{1}{2} J\big(\rho,|\nabla \psi|^2\big))\,d^2r\,,\\
        (\p_t+\mathcal{L}_{\bs{\wh v}})  (Q_W \, d^2r) &= \sigma^2\frac{1}{2}d\rho\wedge d(\wh{w}^2)= \frac12 J\Big(\rho\,,{\frac{\sigma^2 {\wt w}^2}{ \rho^2}} \,\Big)\,d^2r\,.
    \end{split}
    \label{QFW-dyn}
\end{align}
{From the two relations in \eqref{QFW-dyn}, one sees that the buoyancy gradient $\nabla\rho$ couples the PV dynamics of the waves $(Q_W)$ and currents $(Q_F)$, each to their corresponding kinetic energy. In the case of constant buoyancy, $d\rho=0$  in \eqref{QFW-dyn}; so, the PVs of the waves and currents would be separately advected. }

The operator ${\rm div}(\rho\nabla)$ is invertible, so long as $\rho$ is a differentiable positive function, which can be ensured by requiring that this condition holds initially. Consequently, the stream function $\psi$ is related to the other fluid variables by 
\begin{align}
    \psi := ({\rm div}\rho\nabla)^{-1}Q_F \,.\label{psi-def}
\end{align}
The potential vorticity dynamics can then be written in coordinate form as
\begin{align}
\begin{split}
\p_t Q_F + J(\psi, Q_F) 
&=  J\Big(\rho\,,{\frac{1}{2}|\nabla_{\bs{r}}\psi|^2} \,\Big)
\,,\\
\p_t Q_W + J(\psi, Q_W) 
&=  J\Big(\rho\,,{\frac{\sigma^2 {\wt w}^2}{2 \rho^2}} \,\Big)
\,,\\
\hbox{with}\quad 
Q_F := {\rm div}(\rho\nabla\psi)
\quad \hbox{and}&\quad 
Q_W := J\big( \sigma^2 \wt w \,,\, \zeta \big)
\,,\\
\p_t \rho + J(\psi, \rho) &= 0 
\,,\\
\p_t \zeta + J(\psi, \zeta) &= {\wh w} =:{\wt w}/\rho
\,,\\
\p_t (\sigma^2{\wt w}) + J(\psi, \sigma^2{\wt w}) 
&= 
 - \,{\frac{\rho_{ref}  \zeta}{Fr^2}} 
\,.
\end{split}
\label{PV-sys-Bdyn}
\end{align}

\begin{theorem}\label{Ham-thm-rho}
The Legendre transform yields the Hamiltonian formulation of our system of wave-current equations \eqref{PV-sys-Bdyn}, which with ${\wt w}=\rho{\wh w}$ may be written in the untangled block-diagonal Poisson form as
\begin{align}
\begin{split}
\frac{\p}{\p t}
\begin{bmatrix}
Q \\ \rho \\ \sigma^2{\wt w} \\ \zeta 
\end{bmatrix}
= 
\begin{bmatrix}
J( Q,\,\cdot\,) & J(\rho,\,\cdot\,)  & 0 & 0
\\
J(\rho,\,\cdot\,) 	    & 0  & 0 & 0
\\
0	& 0 & 0  & -1
\\
0	 & 0  & 1 & 0
\end{bmatrix}
\begin{bmatrix}
	{\delta h}/{\delta Q} = \psi
	\\
	{\delta h}/{\delta \rho} = {\wt{\varpi}} 
    \\
	{\delta h}/{\delta (\sigma^2{\wt w})} = {\wt w}/\rho + J(\zeta,\psi)
	\\
	{\delta h}/{\delta \zeta} =  -  J(\sigma^2{\wt w},\psi)
	+ { \frac{\rho_{ref} \zeta}{Fr^2} }  
\end{bmatrix}
\,.
\end{split}
	\label{FS-diag-brkt-tau-B}
\end{align}
The energy Hamiltonian $h(Q,\rho,{\wh w},\zeta)$ associated with this system is given by
\begin{align}
\begin{split}
h(Q,\rho,{\widetilde w},\zeta)
=
\int  
    &\frac{1}{2}\Big(Q - J\big(\sigma^2{\wt w},\zeta\big)\Big) 
    ({\rm div}\rho\nabla)^{-1}
    \Big(Q -   J\big(\sigma^2{\wt w},\zeta\big)\Big) 
    \\&\quad 
    +\bigg(\frac{\sigma^2{\wt w}^2}{2\rho^2} 
    + { \frac{\rho_{ref}}{\rho} \frac{\zeta^2 }{2Fr^2} } \bigg)   \rho\,d^2r
\,.
\end{split}
\label{QPV-sys-erg}
\end{align}

\end{theorem}

\begin{theorem}[Casimir functions] The Casimir functions, conserved by the relation 
$\{ C_{\Phi,\Psi},h \}=0$ with any Hamiltonian $h(\bs{M},D)$ for the block-diagonal Lie-Poisson bracket in equation \eqref{FS-diag-brkt-tau-B} are given by
\begin{align}
C_{\Phi,\Psi} := \int \Phi(\rho)+ Q \Psi(\rho)\,d^2r 
	\label{WC-BlockDiag-brkt-Casimirs}
\end{align}
 
\begin{proof}
The Casimirs $C_{\Phi,\Psi}$ for the direct sum of the Lie-Poisson brackets for $Q$ and $\rho$ and canonical Poisson brackets for ${\wt w}$ and $\zeta$ follows by direct verification that the $C_{\Phi,\Psi}$ are conserved for any differentiable functions, $(\Phi,\Psi)$.
\end{proof}

\end{theorem}

\subsection[C$\circ$M equations in the SVE approximation]{C$\circ$M equations in the slowly varying envelope (SVE) approximation}
\label{sec: SVE}

\paragraph{The SVE solutions apply to satellite observations of sea surface waves.}
From the viewpoint of satellite observations, the vertical motion on the sea surface typically oscillates much more quickly than the rate of change of features in the horizontal motion of the ocean surface currents. In this situation, the standard WKB approximation introduces a solution Ansatz for the slowly varying envelope (SVE) of the rapidly oscillating vertical wave elevation in the standard form \cite{B&G1968,GH1996}, 
\begin{align}
\zeta(\bs{r},t) = \Re \bigg(a(\bs{r},t) \exp \Big( \frac{ i\theta(\bs{r},t)}{\epsilon}\Big)\bigg)
\quad\hbox{with}\quad \epsilon  \ll1 
\,.\label{SVE-ansatz}
\end{align} 
The SVE solution Ansatz \eqref{SVE-ansatz} comprises the product of a slowly varying complex amplitude $a(\bs{r},t)\in\mathbb{C}$  multiplied by a rapidly oscillating phase $\theta(\bs{r},t)/\epsilon\in\mathbb{R}$ with $\epsilon\ll1$ in which the phase factor $\theta(\bs{r},t)$ may also vary slowly as a function of the space and time variables, $(\bs{r},t)$. 

Following \cite{GH1996}, let us substitute the SVE solution Ansatz \eqref{SVE-ansatz} into Hamilton's principle in \eqref{HP-WCI} and find the condition on the parameter $\epsilon\ll1$ that will allow higher order wave terms to be neglected. For this, one computes
\begin{align}
\begin{split}
0 = \delta S_{SVE} &= \delta\int_a^b \ell_{SVE}(\bs{\wh{v}},D,\rho; a,\theta)\,dt
\\&= \delta\int_a^b \int_{\cal D} 
\frac{1}{2}D\rho |\bs{\wh{v}}|^2 - p(D-1) 
+ \frac{\sigma^2}{2} D\rho \bigg( \Big( \frac{d\zeta}{dt} \Big)^2  
- \frac{\rho_{ref}}{\rho} \frac{\zeta^2}{2\sigma^2Fr^2}  \bigg) 
\,d^2r\,dt
\\&= \delta\int_a^b \int_{\cal D} 
\frac{1}{2}D\rho |\bs{\wh{v}}|^2 - p(D-1) 
\\&\hspace{2cm}
+ \frac{\sigma^2}{8} D\rho \bigg( \Big| \frac{da}{dt} \Big|^2  
+ \frac{2}{\epsilon}\frac{d\theta}{dt}
\Im \Big( a^*  \frac{da}{dt} \Big)
+ \frac{|a|^2}{\epsilon^2} \Big(  \Big( \frac{d\theta}{dt} \Big)^2 
- \frac{\rho_{ref}}{\rho} \frac{\epsilon^2}{\sigma^2 Fr^2}\Big) \bigg) 
\,d^2r\,dt
\\&\simeq \delta\int_a^b \int_{\cal D} 
\frac{1}{2}D\rho |\bs{\wh{v}}|^2 - p(D-1) 
\\&\hspace{1cm}
+ \frac{\sigma^2|a|^2}{8\epsilon^2} D\rho \bigg(  \big(\partial_t \theta + \bs{\wh{v}} \cdot \nabla_{\bs{r}} \theta  \big)^2 - \frac{\rho_{ref}}{\rho} \frac{\epsilon^2}{\sigma^2 Fr^2} \bigg) 
\,d^2r\,dt
+ O\left(\frac{\sigma^2}{\epsilon}\right)
%
\,.
\end{split}
\label{HP-WCI-SVE}
\end{align}
The leading order wave term {$O(\epsilon^{-2})$ with $\epsilon\ll1$ in Hamilton's principle will dominate the solution and the remaining wave terms in the second line of equation \eqref{HP-WCI-SVE} may be neglected, when \footnote{The ratio $\epsilon^2/(\sigma^2 Fr^2) = O(1)$ is required {for the rate of change of the phase parameter $\theta(\bs{r},t)$ of the SVE wave solution Ansatz \eqref{SVE-ansatz} to match the time scale of the density $\rho(\bs{r},t)$ in equation \eqref{HP-WCI-SVE}.}}}
\begin{align}
\epsilon \ll 1
\,,\quad
\frac{\epsilon^2}{\sigma^2 Fr^2} = O(1), 
\quad\hbox{and}\quad
\sigma^2 Fr^2  \ll1
\,.
\label{HP-WCI-SVE}
\end{align}
According to the estimates in \eqref{FrBV-est} there is a range of physical parameters relevant to satellite 
observations in which the SVE approximation applies, for $\sigma^2 Fr^2  \ll1$. 

To continue the investigation of the SVE description of wave-current interactions on the sea surface, we take variations of the action integral \eqref{HP-WCI-SVE} to find the following set of equations,
\begin{align}
\begin{split}
	\delta \bs{\wh{v}}:&\quad 
	\frac{\delta \ell}{\delta \bs{\wh{v}}} = D\rho \Big( \bs{\wh{v}}\cdot d\bs{r} 
	 + {\cal N} d \frac{d\theta}{dt}  \Big) \otimes d^2r
	 \quad\hbox{with}\quad {\cal N}:= \frac{\sigma^2|a|^2}{4\epsilon^2} 
		\,,\\ 
	\delta |a|^2 :&\quad  
	\frac{\delta \ell}{\delta |a|^2}  =  \frac{\sigma^2}{8 Fr^2} D\rho \bigg( 
	\Big(\frac{d\theta}{dt}\Big)^2 - \frac{\rho_{ref}}{\rho} \bigg) 
	= 0 \quad\hbox{at}\quad O\left(\frac{\sigma^2}{\epsilon^2}\right) 
	\\& 
	\Longrightarrow 
	\frac{d\theta}{dt} =: -\, \omega +  \bs{\wh{v}}\cdot \bs{k} = \pm \frac{\sqrt{\rho\rho_{ref}}}{\rho}
	\quad\hbox{with}\quad \omega(\bs{r},t) = -\p_t \theta
	\quad\hbox{and}\quad
	\bs{k}(\bs{r},t) = \nabla_{\bs{r}}\theta 
		\,,\\ 
	\delta \theta :&\quad 
	\frac{\delta \ell}{\delta \theta} = 0 \Longrightarrow
	\partial_t {\cal A} + \text{div} ({\cal A} \bs{\wh{v}}) = 0\,, \quad\hbox{with}\quad 
	{\cal A} := D\rho{\cal N}\frac{d\theta}{dt} 
	\quad\hbox{and}\quad {\cal N}:= \frac{\sigma^2|a|^2}{4 \epsilon^2} 
	\,,\\
	\delta D:&\quad 
	\frac{\delta \ell}{\delta D } 
	= \frac{\rho}{2}|\bs{\wh{v}}|^2  - p
	\,,\\
	\delta\rho :&\quad 
	\frac{\delta \ell}{\delta \rho } 
	= \frac{D}{2}|\bs{\wh{v}}|^2  
	\,,\\
	\delta p:&\quad 
	D-1 = 0 \implies \text{div}_{\bs{r}}\bs{\wh v} =0 
	\,,\ \hbox{ Hence, }\  \partial_t{\cal A}  + \bs{\wh{v}}\cdot\nabla_{\bs{r}}{\cal A} = 0
	\Longrightarrow \partial_t|a|^2  + \bs{\wh{v}}\cdot\nabla_{\bs{r}}|a|^2 = 0
	\,.
\end{split}
\label{var-derivs-SVE}
\end{align}

In the second line of \eqref{var-derivs-SVE} we see that stationarity of the action integral with respect to variations in $|a|^2$ acts as a Lagrange multiplier to impose a constraint which relates the dynamics of the wave phase $\theta$ to the buoyancy. This constraint relation involves the Doppler-shifted frequency of the waves, as shown in the third line of \eqref{var-derivs-SVE}. In combination with conservation of the wave action density and the divergence free condition on the fluid flow velocity $\bs{\wh{v}}$, this constraint relation implies in the last line of \eqref{var-derivs-SVE} that the wave magnitude $|a|^2$ is advected by the fluid flow. Because of the oscillatory nature of the solution Ansatz \eqref{SVE-ansatz}, the sign of the wave phase in $d\theta/dt = \p_t \theta + \bs{\wh{v}}\cdot\nabla_{\bs{r}}\theta$ in the second line above is immaterial. Hence, hereafter, we will choose the positive root for $d\theta/dt=\sqrt{\rho\rho_{ref}}/\rho$. 

From the conservation of wave action density ${\cal A}$ in \eqref{var-derivs-SVE} and the definitions of the advected fluid variables, one finds that $|a|^2$, $D$ and $\rho$ satisfy the following advection relations
   \begin{align}
   \begin{split}    
        (\p_t+\mathcal{L}_{\bs{\wh v}})(D\,d^2r) = 0 &\Longrightarrow \partial_t D + \text{div}_{\bs{r}}(D\bs{\wh{v}}) 
        =0
        \quad\hbox{with}\quad D=1
        \,,\\ 
        (\p_t+\mathcal{L}_{\bs{\wh v}})\rho = 0 &\Longrightarrow 
        \partial_t\rho  + \bs{\wh{v}}\cdot\nabla_{\bs{r}}\rho = 0       
        \,,\\
        (\p_t+\mathcal{L}_{\bs{\wh v}})|a|^2 = 0 &\Longrightarrow 
        \partial_t |a|^2  + \bs{\wh{v}}\cdot\nabla_{\bs{r}}|a|^2 = 0 
        \,,
        \end{split}
        \label{AdvecQuants-SVE}
    \end{align}
where $\mathcal{L}_{\bs{\wh v}}$ denotes the Lie derivative operation along the horizontal velocity vector field, $\bs{\wh v}$. The Lie derivative notation $\mathcal{L}_{\bs{\wh v}}$ provides coordinate-free brevity in proving the following Kelvin circulation theorem for thermal wave-current theory. 
\begin{theorem}[Kelvin-Noether circulation theorem]\label{theoremEP}
The variational equations in \eqref{var-derivs-SVE} imply the following Kelvin circulation theorem
   \begin{align}
   \frac{d}{dt}\oint_{c(\bs{\wh v})}  \Big( \bs{\wh{v}}\cdot d\bs{r} + {\cal N} d \frac{d\theta}{dt}  \Big)  
   = - \oint_{c(\bs{\wh v})} \frac{1}{\rho}dp\,.
   \label{KNthm-SVE}
   \end{align}
\end{theorem}

\begin{proof}
The Euler-Poincar\'e (EP) theorem \cite{HMR1998} in this case yields
   \begin{align}
   \begin{split}        
 (\p_t+\mathcal{L}_{\bs{\wh v}})\frac{\delta \ell}{\delta \bs{\wh{v}}} 
 &= 
 \frac{\delta \ell}{\delta D } \diamond D + \frac{\delta \ell}{\delta \rho } \diamond \rho
 :=  D \nabla_{\bs{r}}\frac{\delta \ell}{\delta D } - \frac{\delta \ell}{\delta \rho } \nabla_{\bs{r}} \rho
 \,.
        \end{split}
        \label{EPeqn-SVE1}
    \end{align}
Here, the diamond $(\,\diamond\,)$ operator is defined for a fluid advected quantity $f$ by 
   \begin{align}
   \Big\langle  \frac{\delta \ell}{\delta f } \diamond f \,,\, X \Big\rangle_{\mathfrak{X}}
   =:     \Big\langle  \frac{\delta \ell}{\delta f }  \,,\, - \pounds_X f \Big\rangle_V
          \,.\label{EPeqn-SVE2}
    \end{align}
In \eqref{EPeqn-SVE2}, $X\in \mathfrak{X}(\bb{R}^2)$ is a (smooth) vector field defined on $\bb{R}^2$ and $f\in V$ is a vector space of advected quantities. These advected quantities are the scalar function, $\rho$, and the areal density, $D\,d^2r$. 

Upon using the advection relations for $D$ and $\rho$ in \eqref{AdvecQuants-SVE} and the corresponding variational derivatives in \eqref{var-derivs-SVE}, the EP equation in \eqref{EPeqn-SVE1} simplifies to 
   \begin{align}
   \begin{split}        
 (\p_t+\mathcal{L}_{\bs{\wh v}})&\Big(\frac{1}{D\rho}\frac{\delta \ell}{\delta \bs{\wh{v}}} \Big)
 =   \frac{1}{\rho} \nabla_{\bs{r}}\frac{\delta \ell}{\delta D } - \frac{1}{D\rho}\frac{\delta \ell}{\delta \rho } \nabla_{\bs{r}} \rho
\,. \\   \hbox{Equation \eqref{var-derivs-SVE} then yields} \quad
(\p_t+\mathcal{L}_{\bs{\wh v}})&\Big( \bs{\wh{v}}\cdot d\bs{r} + {\cal N} d \frac{d\theta}{dt}  \Big) 
 =
     - \rho^{-1}dp + d\Big(\frac{1}{2}|\bs{\wh{v}}|^2\Big) 
  \,.
        \end{split}
        \label{EPeqn-SVE}
    \end{align}
Inserting the last relation into the following standard relation for the time derivative of a loop integral then completes the proof of equation \eqref{KNthm-SVE} appearing in the statement of the theorem,
    \begin{align}
     \frac{d}{dt}\oint_{c(\bs{\wh v})}  \Big( \bs{\wh{v}}\cdot d\bs{r} + {\cal N} d \frac{d\theta}{dt}  \Big)
     = \oint_{c(\bs{\wh v})} (\p_t+\mathcal{L}_{\bs{\wh v}}) \Big( \bs{\wh{v}}\cdot d\bs{r} + {\cal N} d \frac{d\theta}{dt}  \Big)
     = \oint_{c(\bs{\wh v})} - \rho^{-1}dp + d\Big(\frac{1}{2}|\bs{\wh{v}}|^2\Big)
     \,.
        \label{EPeqn-Kel}
    \end{align}
Note, however, that equations \eqref{var-derivs-SVE} imply the following combination of advected quantities,
    \begin{align}
    (\p_t+\mathcal{L}_{\bs{\wh v}})\left( {\cal N} d \frac{d\theta}{dt}\right) 
    =
    \frac{\sigma^2}{4 Fr^2} (\p_t+\mathcal{L}_{\bs{\wh v}})\left( |a|^2 d \sqrt{\frac{\rho_{ref}}{\rho}} \right) = 0
    \,.
    \label{Wave-Kel}
    \end{align}
Consequently, the wave-momentum 1-form ${\cal N} d (\frac{d\theta}{dt})$ is advected by the fluid flow and the Kelvin circulation theorem in equation \eqref{EPeqn-Kel} reduces to the standard circulation theorem for the 2D Euler fluid equations. 
\end{proof}

\begin{remark}[Separation of wave and current motion in the SVE approximation]\label{Remark:nonacceleration SVE}
    The decoupling of the Kelvin-Noether circulation theorem into its wave and current components for the SVE approximation is inherited from the un-approximated model. When modifications to the un-approximated model which removes this property are added, one would expect the new SVE approximation to lose the non-acceleration result. 
\end{remark}
\color{black}
\begin{remark}
Equation \eqref{Wave-Kel} implies advection of the 1-form $|a|^2 d \rho$, which in turn implies advection of the Jacobian $J(|a|^2,\rho)$. Since the fluid flow is area preserving, $\textrm{div}\bs{\wh{v}}=0$,  the following 2-form will also be advected,  
\begin{align}
    \big(\partial_t   + \bs{\wh{v}}\cdot\nabla_{\bs{r}}\big)\Big(d|a|^2\wedge d\rho\Big) = 0
\,.
\label{Jac-advec}
\end{align}
Thus, the divergence-free flow of $\bs{\wh{v}}$ preserves the area element $d|a|^2\wedge d\rho$. This means that if the gradients $\nabla|a|^2$ and $\nabla\rho$ are not aligned initially, then they will remain so. It also means that equilibrium solutions of \eqref{Jac-advec} will be symplectic manifolds \cite{Holm-GM1text-2011}. 
\end{remark}

After deriving these equations, one may finally evaluate the constraint $D=1$ imposed by the variation in pressure $p$ to 
obtain further simplifications.

\subsection[Thermal potential vorticity dynamics with SVE]{Thermal potential vorticity dynamics with SVE on a free surface}
The momentum map arising from the variations of the action in \eqref{var-derivs-SVE} is given by 
\begin{align}
\begin{split}
	\frac1D \frac{\delta \ell}{\delta \bs{\wh{v}}} &= \rho \Big( \bs{\wh{v}}\cdot d\bs{r} 
	 + {\cal N} d \frac{d\theta}{dt}  \Big) 
	 \quad\hbox{with}\quad {\cal N}:= \frac{\sigma^2N^2|a|^2}{4} =: \Gamma |a|^2
	 \quad\hbox{and}\quad \frac{d\theta}{dt} =  \sqrt{\frac{\rho_{ref}}{\rho}}
	 \\ \hbox{so}\quad
	 \frac1D \frac{\delta \ell}{\delta \bs{\wh{v}}} 
	 &= \rho \Big( \bs{\wh{v}}\cdot d\bs{r} + \Gamma |a|^2 d (\sqrt{\rho\rho_{ref}}/\rho) \Big) 
\end{split} \label{Q-def-SVE}
\end{align}
According to the Euler-Poincar\'e equation \eqref{EPeqn-SVE}, the dynamics of the fluid and wave components of the 1-form in \eqref{Q-def-SVE} \emph{separates} into the following equations,
\begin{align}
\begin{split}
(\p_t+\mathcal{L}_{\bs{\wh v}}) \Big( \rho \big( \bs{\wh{v}}\cdot d\bs{r} \big)\Big)
 &=
     - dp + \frac{\rho}{2} d\big(|\bs{\wh{v}}|^2\big) 
 \\
 (\p_t+\mathcal{L}_{\bs{\wh v}}) \big( |a|^2 d \sqrt{\rho\rho_{ref}} \big)
 &= 0
  \,.
\end{split}
\label{momap-SVE}
\end{align}
This means that the mass-weighted thermal potential vorticity (TPV) dynamics also separates into the following fluid and wave components, $Q=Q_F + Q_W$, given by
\begin{align}
\begin{split}
 Q \, d^2r &:= d  \bigg(\rho \Big( \bs{\wh{v}}\cdot d\bs{r} 
	 + \Gamma |a|^2 d \sqrt{\frac{\rho_{ref}}{\rho}} \Big)\bigg)
\\&=
\Big({\rm div}(\rho\nabla\psi) - \Gamma J\Big(|a|^2, \sqrt{\rho\rho_{ref}} \Big)\,d^2r
\quad\hbox{when}\quad
\bs{\wh v} =\nabla^\perp\psi 
\quad\hbox{for}\quad D=1
\,,\\&=
 Q_F \, d^2r + Q_W \, d^2r
\,,\\
\hbox{with}\quad 
 Q_F &:= {\rm div}(\rho\nabla\psi)
\quad \hbox{and}\quad 
 Q_W := \Gamma J\big( \sqrt{\rho\rho_{ref}} \,,\, |a|^2 \big)
\,.
\end{split}
\label{PV-def-SVE}
\end{align}
Then, again, the differentials of the separate equations in \eqref{momap-SVE} yield the `non-acceleration' result, 
\begin{align}
\begin{split}
(\p_t+\mathcal{L}_{\bs{\wh v}})  (Q_F \, d^2r)
 &=
      \frac{1}{2}d\rho\wedge d|\bs{\wh{v}}|^2 
  = \frac{1}{2} J\big(\rho,|\nabla \psi|^2\big)\, d^2r
  \,,
\\
(\p_t+\mathcal{L}_{\bs{\wh v}})  (Q_W \, d^2r)
	 &= 0 
\end{split}
\label{Q-PV-SVE-split}
\end{align}
Equivalently, in coordinates one has 
\begin{align}
\begin{split}
\p_t Q_F + \bs{\wh v}\cdot \nabla Q_F &= \frac{1}{2} J\big(\rho,|\nabla \psi|^2\big)
\,,\\
\p_t Q_W + \bs{\wh v}\cdot \nabla Q_W &= 0
\,,\\
\hbox{with}\quad 
 Q_F := {\rm div}(\rho\nabla\psi)
\quad \hbox{and}&\quad 
 Q_W := \Gamma J\big( \sqrt{\rho\rho_{ref}} \,,\, |a|^2 \big)
\,,\\
\partial_t \rho  + \bs{\wh{v}}\cdot\nabla_{\bs{r}} \rho &= 0
\quad\hbox{and}\quad
\Gamma = \frac{\sigma^2}{4 Fr^2} = O(1)
\,,\\
\partial_t|a|^2  + \bs{\wh{v}}\cdot\nabla_{\bs{r}}|a|^2 &= 0
\,,\\
\partial_t \theta  + \bs{\wh{v}}\cdot\nabla_{\bs{r}} \theta &=  \frac{\sqrt{\rho\rho_{ref}}}{\rho}
\,.\end{split}
\label{Q-PV-SVE-eqn}
\end{align}
The operator $({\rm div}\rho\nabla)$ is invertible, so long as $\rho$ is a differentiable positive function, which can be ensured by requiring that this condition holds initially, since $\rho$ is advected. Consequently, the stream function $\psi$ is related to the other fluid variables by 
\begin{align}
\psi := ({\rm div}\rho\nabla)^{-1}Q_F 
\,.\label{psi-def}
\end{align}
The dynamics of the equation set \eqref{Q-PV-SVE-eqn} explains why the  various physical components of the flow coordinate their movements, as seen in satellite observations in figure \ref{fig:coherence}. In particular, the motion of buoyancy $\rho$ and squared wave amplitude $|a|^2$ are coordinated with each other through the advection of the momentum 1-form $|a|^2 d\rho$ and the area 2-form $d|a|^2\wedge d\rho$. Likewise the the motion of the fluid potential vorticity $Q_F$ and the mass density  $\rho$ are coordinated with each other through the mass-weighted definition of the stream function in \eqref{psi-def}. These considerations emphasise again the importance of horizontal buoyancy gradients in sea surface dynamics.

\section{Numerical implementation}\label{SimSpecs}
Our implementation of the C$\circ$M equations \eqref{PV-sys-Bdyn} and the C$\circ$M equations in the SVE approximation \eqref{Q-PV-SVE-eqn} used the finite element method (FEM) for the spatial variables. The FEM algorithm we used is based on the algorithm formulated in \cite{HLP2021} and is implemented using the Firedrake \footnote{\url{https://firedrakeproject.org/index.html}} software. In particular, for \eqref{PV-sys-Bdyn} we approximated the fluid potential vorticity $Q_F$, buoyancy $\rho$, wave elevation $\zeta$ and bouyancy weighted wave vertical velocity $\tilde{w}$ using a first order discrete Galerkin finite element space. Similarly, for \eqref{Q-PV-SVE-eqn}, we approximated $Q_F$, $\rho$, square of the wave amplitude $|a|^2$ and wave phase $\theta$ using a first order discrete Galerkin finite element space. The stream function $\psi$ for both models was approximated by using a first order continuous Galerkin finite element space. For the time integration, we used the third order strong stability preserving Runge Kutta method \cite{Got05}. 

Figures \ref{fig:snapshot wcifs} and \ref{fig:snapshot full wcifs} present snapshots of high resolution runs of the C$\circ$M equations and the C$\circ$M equations in the SVE approximation. These simulations were run with the following parameters. The domain is $[0,1]^2$ at a resolution of $512^2$. The boundary conditions are periodic in the $x$ direction, and homogeneous Dirichlet for $\psi$ in the $y$ direction. To see the effects of the waves on the currents, the procedure was divided into two stages for both set of equations. The first stage was performed without wave activity for $T_{spin} = 100$ time units starting from the following initial conditions
\begin{align}
\begin{split}
    Q_F(x,y, 0) & = \sin(8\pi x)\sin(8\pi y) + 0.4\cos(6\pi x)\cos(6\pi y) + 0.3\cos(10\pi x)\cos(4\pi y) +\\ & \qquad 0.02\sin(2\pi y) + 0.02\sin(2\pi x)\,,\\
    \rho(x,y, 0)& = 1 + 0.2\sin(2\pi x)\sin(2\pi y)\quad\hbox{and}\quad {\rho_{ref}=1}\,.
\end{split}
\end{align}
The purpose of the first stage was to allow the system to spin up to a statistically steady state without any wave activity. The PV and buoyancy variables at the end of the initial spin-up period are denoted as $Q_{spin}(x,y) = Q_F(x,y,T_{spin})$ and $\rho_{spin}(x,y) = \rho(x,y,T_{spin})$. Figures of these variables are shown in figure \ref{fig:current spinup state}.
\begin{figure}[h!]
    \centering
    \begin{subfigure}[b]{0.45\textwidth}
		\centering
		\includegraphics[width=\textwidth, height=\textwidth]{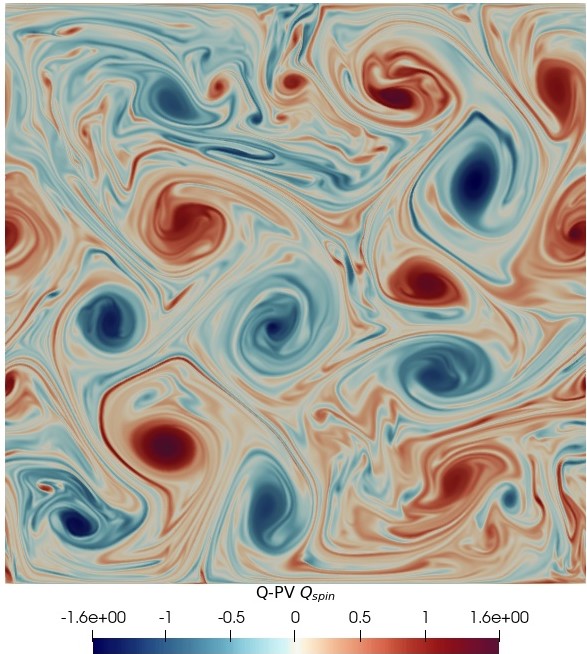}
	\end{subfigure}
	\begin{subfigure}[b]{0.45\textwidth}
		\centering
		\includegraphics[width=\textwidth, height=\textwidth]{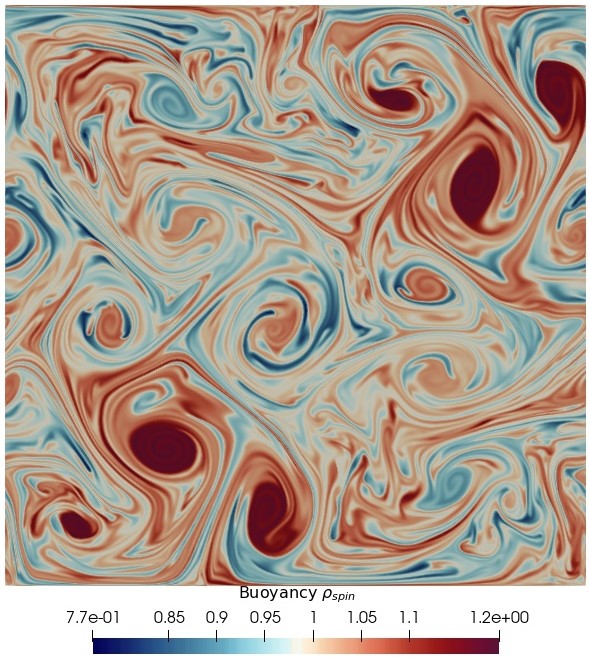}
	\end{subfigure}
	\caption{These figures show the results of the first stage of the simulation in which only fluid motion is present and the wave degrees of freedom are absent. The panels show fluid potential vorticity $Q_F$ (left) and buoyancy $\rho$ (right). The fluid state obtained from the first stage was used as the initial condition for the second stage simulations in which wave variables were included. These distributions of fluid properties show strong spatial coherence. The coordination of wave and fluid properties that emerges in the second stage of the simulations shown in Figure \ref{fig:snapshot wcifs} and \ref{fig:snapshot full wcifs} arises from the interaction between the wave and current components of the flow which is mediated by the buoyancy gradient.}
	\label{fig:current spinup state}
\end{figure}
In the second stage, the full simulations including the wave variables were run with the initial conditions for the flow variables being the state achieved at the end of the first stage. To start the second stage for \eqref{PV-sys-Bdyn}, wave variables were introduced with the following initial conditions
\begin{align}
    \begin{split}
        \zeta(x,y,0) & = \sin(8\pi x)\sin(8\pi y) + 0.4\cos(6\pi x)\cos(6\pi y) + 0.3\cos(10\pi x)\cos(4\pi y) +\\ 
        & \qquad 0.02\sin(2\pi y) + 0.02\sin(2\pi x)\,,\\
        \tilde{w}(x,y,0) &= 0\,,\quad Q_F(x,y,0) = Q_{spin}(x,y)\,,\quad \rho(x,y,0) = \rho_{spin}(x,y)\,,\\
        \sigma^2 Fr^2 &= 10^{-2}\,.
    \end{split}
\end{align}
To start the second stage for \eqref{Q-PV-SVE-eqn}, wave variables were introduced with the following initial conditions
\begin{align}
\begin{split}
    |a|^2(x,y,0) & = \big(\sin(8\pi x)\sin(8\pi y) + 0.4\cos(6\pi x)\cos(6\pi y) + 0.3\cos(10\pi x)\cos(4\pi y) +\\ 
    & \qquad 0.02\sin(2\pi y) + 0.02\sin(2\pi x)\big)^2\,,\\
    \theta(x,y,0) &= 0\,,\quad Q_F(x,y,0) = Q_{spin}(x,y)\,,\quad \rho(x,y,0) = \rho_{spin}(x,y)\,.
\end{split}
\end{align}
\begin{remark}\label{theta-rho}
Importantly, the wave phase $\theta$ in the second stage was set initially to zero. Thereafter, the wave phase $\theta$ increased linearly in time in proportion to the advected quantity $\sqrt{\rho\rho_{ref}}/\rho$ following each flow line, as implied by the last equation in \eqref{Q-PV-SVE-eqn}. 
\end{remark}


\section{Conclusion and Outlook}
This paper models the effects of thermal fronts on the dynamics of the ocean's waves and currents. It introduces and simulates two models of thermal wave-current dynamics on a free surface. The original C$\circ$M model is derived from Hamilton's principle via the composition of two maps which represent the horizontal and vertical motion respectively. The second, a slowly varying encelope (SVE) model, is introduced via the standard WKB approximation which takes advantage of large separation of the space-time scales between the slow horizontal currents and fast vertical oscillations. In particular, the second model introduces the WKB solution Ansatz into Hamilton's principle, whereupon the time integral averages over the phases of the rapid oscillations that are out of resonance with the slowly varying envelope. Model runs of both models are presented in which the buoyancy mediates the dynamics of the currents and waves, as seen in Figures \ref{fig:snapshot wcifs} and \ref{fig:snapshot full wcifs}. These simulations also validate the use of the WKB approximation for two reasons. First, the resolved small scale wave features of the original C$\circ$M model lie primarily within the envelope defined by the SVE approximate model. This means that the dynamics of the spatial features of the SVE approximate model are consistent with those of the original C$\circ$M model, although the resolved space and time scales differ. Secondly, requiring that $\epsilon^2/(Fr^2\sigma^2)= O(1)$ ensures that the time scale for the wave envelope dynamics matches that for the fluid motion.  

Nonetheless, the two models introduced here merit further study in several directions. For example, it remains to: (1) quantify the correlations observed visually; (2) determine their rate of formation; and (3) parameterise the model for comparison and analysis of the satellite data on which their derivations were based. Furthermore, the models discussed here involve only variables that are evaluated on the free surface and therefore they neglect bathymetry. A scientific challenge persists in understanding regions of the ocean where bathymetry has profound effects on the observable surface dynamics, such as in the Lofoten vortex \cite{VKL2015}. This is a multiscale issue that might be addressed by including mesoscale modulations of the sub-mesoscale models derived here. One candidate for providing the mesoscale modulations would the thermal quasi-geostrophic (TQG) model in which bathymetry has recently been included \cite{HLP2021}. 

{The currents are modelled here by the two dimensional incompressible Euler equations, as seen in equations \eqref{HP-WCI-A-Eul-Lag} and \eqref{Eul-eqn}. Incompressibility is a reasonable assumption in some regions of the ocean, for example when the quasigeostrophic approximation is valid. There are regions in the upper ocean where other equations are more suitable for modelling currents, and the development and investigation of such two dimensional models is an open problem which warrants further consideration.}

As mentioned in Remark \ref{Remark:nonacceleration}, the wave component of the model presented here does not create circulation in the currents. The instabilities present in satellite simulations indicate that additional modelling is needed to fully capture this effect. Future work will investigate approaches for modelling these instabilities.

Many other questions remain about wave-current interaction. The full extent of submesoscale ocean dynamics is by no means adequately described by existing models. For example, we have little understanding of the formation and dynamics of various sea-surface phenomena, including the so-called `spirals on the sea' \cite{MunkArmiFZ2000}. Other questions are emerging because the ocean has absorbed in excess of 90\% of the heat present in the earth system as a result of human activity during the post-industrial era \cite{IPCC2019}. The absorption of heat from the warming atmosphere is ongoing and it is forecast to become more dramatic. This absorption has resulted in `marine heat waves', which are predicted to increase in frequency and severity. These changes to the upper ocean, where most of this heat is stored, could have a profound effect on the dynamical landscape of our oceans. These effects may, in turn, influence our weather and climate systems. Over the millennia, the ocean has approached statistical equilibrium under its current forcing conditions. Using modelling terminology, one says the ocean is well `spun-up'. However, the continued warming of the ocean is likely to influence the number and intensity of thermal fronts. One hopes that mathematical models will provide a useful framework for estimating some of the potential impacts of these thermal fronts on atmospheric effects, as well.


\subsection*{Acknowledgements} 
We are grateful to our friends and colleagues who have generously offered their time, thoughts, and encouragement in the course of this work during the time of COVID-19. Thanks to A. Arnold, B. Chapron, D. Crisan, E. Luesink,  A. Mashayekhi, and J. C. McWilliams for their thoughtful comments and discussions. Particular thanks to B. Chapron, {for extensive discussions of satellite oceanography and for providing the satellite data in  figures \ref {fig:snapshot Chapron} and \ref{fig:coherence}. We also thank B. Fox-Kemper for constructive discussions of modelling approaches in physical oceanography. These discussions helped us clarify the distinction between the present C$\circ$M modelling approach and the classical balance equation approach}. The authors are grateful for partial support, as follows. 
DH for European Research Council (ERC) Synergy grant STUOD - DLV-856408;
RH for the EPSRC scholarship (Grant No. EP/R513052/1); and 
OS for the EPSRC Centre for Doctoral Training in the Mathematics of Planet Earth (Grant No. EP/L016613/1).

\end{document}